\documentclass[oneside,english]{amsart}
\usepackage[T1]{fontenc}
\usepackage[latin9]{inputenc}
\usepackage{geometry}
\usepackage{amssymb}
\usepackage{graphicx,color}

\makeatletter
\numberwithin{equation}{section} 
\numberwithin{figure}{section} 
  \theoremstyle{plain}
  \newtheorem*{thm*}{Theorem}
  \theoremstyle{plain}
  \newtheorem*{lem*}{Lemma}
  \theoremstyle{plain}
  
  \theoremstyle{remark}
  \newtheorem*{acknowledgement*}{Acknowledgement}

\usepackage{babel}
\makeatother

\begin{document}

\title{TransverseDiff gravity is to scalar-tensor as unimodular gravity
is to General Relativity.}

\author{J.J. Lopez-Villarejo}

\date{4th January, 2011}

\address{European Organization for Nuclear Research, CERN CH-1211, Geneve 23, Switzerland
\and \newline Departamento de F\'{i}sica Te\'{o}rica and Instituto de F\'{i}sica Te\'{o}rica UAM/CSIC, E-28049-Madrid, Spain}

\email{jj.lopezvillarejo@cern.ch; jjlopezvillarejo@gmail.com}

\begin{abstract}
Transverse Diffeomorphism (TDiff) theories are well-motivated theories
of gravity from the quantum perspective, which are based upon a gauge
symmetry principle. The main contribution of this work is to firmly
establish a correspondence between TransverseDiff and the better-known
scalar-tensor gravity --- in its more general form ---, a relation
which is completely analogous to that between unimodular gravity and
General Relativity. We then comment on observational aspects of TDiff.
In connection with this proof, we derive a very general rule that
determines under what conditions the procedure of fixing a gauge symmetry
can be equivalently applied before the variational principle leading
to the equations of motion, as opposed to the standard procedure,
which takes place afterwards; this rule applies to gauge-fixing terms
without derivatives. 

\end{abstract}
\maketitle

\section{Introduction}

Symmetry principles have a great relevance in modern Physics theories.
Specifically, there is a widespread point of view on General Relativity (GR),
stemming from its possible quantum interpretation, that considers
it as the theory of spacetime diffeomorphism (Diff) invariance for
spin-two massless particles: one starts from a linearized gravity
action for a symmetric rank-two tensor in Minkowski space and, under
the guidance of the Diff symmetry principle, one is ``naturally''
led to the full Einstein theory in the non-linear regime \cite{BoulwareDeser:1974sr, Feynman:lect_gravit}
(although this has been challenged in \cite{Padmanabhan:2004xk}). Under this view, the \emph{external}
symmetry group of (active) diffeomorphisms acts as a \emph{gauge}
group which proves essential for eliminating the unphysical ill-defined
polarization modes (ghosts) of the massless graviton. This perspective on gravity deviates from Einstein's geometrical intuition, which is recovered as a by-product. For our purposes in particular, the metric (field) $g_{\mu \nu}$ may be better understood as simply a rank-two tensor (field), and the covariant derivatives that will appear are convenient book-keeping forms for expressions involving the Christoffel symbols of $g_{\mu \nu}$.

On the other hand, it was already pointed out that the necessary and
sufficient symmetry group for a consistent description of the massless
graviton is not the full group of diffeomorphisms (Diff(M)) but rather
a maximal subgroup of it \cite{BijDamNg:1981ym}. Transformations
belonging to this subgroup have been dubbed \emph{transverse} (TDiff(M))
\cite{AlvarezGarriga:2006uu,AlvarezAnton:2007nn}, since the parameter
describing the gauge transformation at this linear level is transverse,
\emph{$\partial_{\mu}\xi^{\mu}=0$}, such as in the (infinitesimal) transformation
law for the linearized metric:\begin{equation}
h_{\mu\nu}(x)\rightarrow h_{\mu\nu}(x)+\partial_{\mu}\xi_{\nu}(x)+\partial_{\nu}\xi_{\mu}(x),\,\, with\,\,\partial_{\mu}\xi^{\mu}(x)=0.\label{eq:linear TDiff transformation}\end{equation}
At the non-linear level (the analogous to GR), one can maintain consistently the same restriction \emph{$\partial_{\mu}\xi^{\mu}=0$}
on the symmetry group\footnote{To go from the linear to the non-linear  theory, one could try to repeat the procedure in \cite{BoulwareDeser:1974sr}. This has been attempted elsewhere unsuccessfully and, moreover, there are serious concerns about the uniqueness of the procedure already for the case of GR \cite{Padmanabhan:2004xk}. Instead, we look for a consistent and ``logical'' generalization of the linear constraint on the gauge parameter in \eqref{eq:linear TDiff transformation}, which has the latter as its linear limit: this non-linear constraint turns out to be exactly the same. The other ``most sensible''
generalization, where $\xi^{\mu}$ is restricted with $\nabla_{\mu}\xi^{\mu}=0$
(and $g$ does not transform at all, not even as a scalar does) is
doomed since one cannot accommodate a kinetic term for $g$: $\frac{1}{2}f_{k}(g)g^{\mu\nu}\partial_{\mu}g\partial_{\nu}g$ does not transform as a scalar.%
} so that, for example, the (infinitesimal) transformation of the metric reads:\begin{equation}
g_{\mu\nu}(x)\rightarrow g_{\mu\nu}(x)+\nabla_{\mu}\xi_{\nu}(x)+\nabla_{\nu}\xi_{\mu}(x),\,\,\,\,\, with\,\,\partial_{\mu}\xi^{\mu}(x)=0.\label{eq:non-linear TDiff transformation}\end{equation}
For a finite transformation, this corresponds to the subgroup of (finite)
general coordinate transformations with Jacobian determinant equal
to unity.

Following under this {}``quantum view'', consider that the true
symmetry of nature was TDiff, instead of Diff. Then, one is impelled
to construct the most general action compatible with TDiff symmetry.
The result is that, since $g\equiv-det(g_{\mu\nu})$ behaves as a
scalar under this subgroup, one can add arbitrary functions of $g$
to the usual terms of the GR Lagrangian \cite{AlvarezAnton:2007nn}
(it may be useful for the reader to consult Appendix A for what follows).
Therefore, the most general Lagrangian invariant under transverse
diffeomorphisms and of second order in derivatives reads:\begin{equation}
S_{\text{TDiff}}=-\frac{1}{2\kappa^{2}}\int d^{4}x\sqrt{g}\left[f\left(\sqrt{g}-1\right)R+2f_{\lambda}\left(\sqrt{g}-1\right)\Lambda+\frac{1}{2}f_{k}\left(\sqrt{g}-1\right)g^{\mu\nu}\partial_{\mu}\sqrt{g}\partial_{\nu}\sqrt{g}\right]+S_{M}\label{eq:TDiff action}\end{equation}
and the matter action may be taken to be of the form\begin{equation}
S_{M}=\int d^{4}x\sqrt{g}L_{SM}\left[\psi_{m},g_{\mu\nu};\,\sqrt{g}-1\right]\label{eq:TDiff matter action}\end{equation}
where we allow for an arbitrary extra-dependence on $g$ in the conventional
Standard Model Lagrangian%
\footnote{For example, the Electromagnetic Lagrangian is generalized to $-\frac{1}{4}f_{EM}(\sqrt{g}-1)F_{\mu\nu}F^{\mu\nu}$%
}; also, we have expressed this dependence on $\sqrt{g}-1$ instead
of $g$ for later convenience. As it stands, the action given by \eqref{eq:TDiff action},
\eqref{eq:TDiff matter action} is clearly non-covariant, but is formulated
in a specific set of coordinates --- up to TDiff transformations of
coordinates. The quantum ultraviolet behavior of such a theory was
discussed in \cite{AlvarezAntonVillarejo:2008zw}; its observational
signatures were addressed in \cite{AlvarezAntonVillarejo:2009ga}
--- we will have some comments to add under the light shed by our
findings. 

This way of proceeding from the TDiff symmetry principle is not, however,
the standard route that followed from the work of van der Bij, van Dam \& Ng \cite{BijDamNg:1981ym}.
Indeed, these same authors started to develop what is known as \emph{unimodular
gravity} (recent proposals include \cite{Earman:2003zj,Smolin:2009ti}).
We compare \& contrast these two approaches in section 2. 

On another note, there is a significant lesson in the work by Kretschmann
\cite{Kretschmann:1917}, when objecting to Einstein's views: every
spacetime theory admits a generally covariant representation (see
also \cite{Norton:2003cx,Nortoneightdecades:1993}). For TDiff, a
representation of this kind can be found, e.g., in \cite{Pirogov:2009hr}.
It introduces an absolute prior-geometric object \cite{LeeFoundations:1973},
namely a background density $\tilde{g}$, so that all occurrences
of $g$ in the scalar Lagrangian%
\footnote{Note that we have implicitly defined the Lagrangian without the usual
factor $\sqrt{g}$.%
} appear in the covariant form $g/\tilde{g}$ (then again, when the
coordinates employed are such that $\tilde{g}$ equals the unit scalar
density we recover straightforwardly our action \eqref{eq:TDiff action}).
In this work, we go one step further to search for a covariant version of TDiff theory. We find that, under certain simple boundary conditions,
TDiff theory is equivalent to a \emph{general} scalar-tensor theory with a specific gauge-fixing. Moreover, we claim that  --- boundary conditions provided --- these two (set of) theories are ``physically equivalent", to the same
extent that General Relativity constrained to the harmonic gauge ---
or any other gauge --- corresponds in its physical consequences to
General Relativity in its covariant standard form (we are just {}``fixing
the gauge''):
\[
\delta\left(\int d^{4}x\sqrt{g}R\right)=0\,\,\,\,\,\,\underset{"physically\, equivalent"}{\Longleftrightarrow}\left[\delta\left(\int d^{4}x\sqrt{g}R\right)=0\right]_{\partial_{\mu}(\sqrt{g}g^{\mu\nu})=0}\]
By \emph{general} scalar-tensor theory we mean that the gravitational
scalar is arbitrarily present in all terms of the Lagrangian. This contrasts with the standard presentation of scalar-tensor theory \cite{Jordan:1955,Fierz:1956zz,BransDicke:1961sx},
which employs the principle of \emph{universal coupling} in the matter
Lagrangian, or the equivalent concept of a \emph{metric theory} of
gravity (see for example \cite{LeeFoundations:1973,WillLivingRev:2006}), and where the occurrence of the gravitational scalar in the matter Lagrangian is limited to a certain type of terms. In passing, these prescriptions were developed to comply with the (Weak) Equivalence Principle in this theory; see however \cite{Damour:2001fn} regarding their convenience.

The idea of a correspondence between TDiff and general scalar-tensor
gravity has already been touched upon in previous studies \cite{AlvarezAntonVillarejo:2008zw,AlvarezAntonVillarejo:2009ga},
and has even made use of, but now the aim is to definitively establish
it; we address it in section 3. 

In sections 4 and 5, we deal with particular aspects of this correspondence:
the energy-momentum tensors and flat limits in both theories. Finally,
we conclude in section 6.

\section{TDiff Gravity vs. Unimodular theories}

The TransverseDiff symmetry at the linear level can be implemented
in two different fashions in the Lagrangian, which give rise to two
very different theories at the non-linear level.

The first route was explored in the aforementioned seminal paper \cite{BijDamNg:1981ym}.
One should note that the transversality condition in \eqref{eq:linear TDiff transformation}
is directly affecting the transformation rule of the trace of $h_{\mu\nu}$:
$h_{\mu}^{\mu}\rightarrow h_{\mu}^{\mu}+\partial_{\mu}\xi^{\mu}$,
such that this quantity is an invariant. Therefore, one can decide
to restrict the value of the trace in the theory to have a fixed value,
typically $h_{\mu}^{\mu}=0$. In the full non-linear regime, this
naturally translates into the condition $g\equiv-det(g_{\mu\nu})=1$,
which explains the name \emph{unimodular gravity}, but one could also
fix the determinant of the metric to any arbitrary function. 

Therefore, unimodular gravity is based upon a reduction of the functional
space on which the Einstein-Hilbert action is defined:\[
g_{\mu\nu}\rightarrow\hat{g}_{\mu\nu},\,\, with\, det(\hat{g}_{\mu\nu})=\epsilon_{0}\]
\begin{equation}
S_{EH}[g_{\mu\nu}]\equiv-\frac{1}{2\kappa^{2}}\int d^{4}x\sqrt{g}R[g_{\mu\nu}]\rightarrow S_{UG}[\hat{g}_{\mu\nu}]\equiv-\frac{1}{2\kappa^{2}}\int d^{4}x\epsilon_{0}R[\hat{g}_{\mu\nu}]\label{eq:unimodular gravity}\end{equation}
where $\epsilon_{0}$ is some fixed scalar density, usually taken
unity. The symmetries of this action are the {}``volume preserving
diffeomorphisms (VPD)\textquotedblright{}, which respect $g(x)=g'(x)$
(and at the infinitesimal level $\delta_{\xi}g=\nabla_{\mu}\xi^{\mu}=\partial_{\mu}(\sqrt{g}\xi^{\mu})=0$).
Note that these are not the same as \eqref{eq:non-linear TDiff transformation}
unless $\epsilon_{0}(x)=const$. We have 9 e.o.m.'s $\delta S/\delta\hat{g}_{\mu\nu}=0$
(of which only 6 are independent, in virtue of the 3 functional-parameter VPD symmetry) for
9 functions $\hat{g}_{\mu\nu}$: the result is a 3 functional-parameter space of solutions. 
This theory, as it turns out, resembles General Relativity with a gauge-fixing constraint on the metric. 
In fact, the only difference at the classical level is an arbitrary integration constant that appears,
and which plays the role of a cosmological constant (see Appendix
B).

On the other hand, the approach that we are interested in, named {}``TDiff
gravity'' after the symmetry, is{\small{} }based upon application
of a symmetry principle on the full space of metrics $g_{\mu\nu}$
(see \eqref{eq:TDiff action} and \eqref{eq:non-linear TDiff transformation}).
There are 10 e.o.m.'s: $\delta S/\delta g_{\mu\nu}=0$ (only 7 independent
in virtue of TDiff-based Bianchi identities), for 10 functions $g_{\mu\nu}$.
This theory is connected to scalar-tensor theory with a gauge-fixing
constraint, as we will see.

\section{Correspondence between TDiff gravity and a general scalar-tensor
theory.}

\begin{thm*}
Of correspondence between scalar-tensor and TDiff gravity.

Let us have a TDiff theory of the general form \eqref{eq:TDiff action}\eqref{eq:TDiff matter action},
abbreviated as \[
S_{\text{TDiff}}=\int d^{4}x\sqrt{g(x)}\mathcal{L}(g_{\mu\nu}(x),\sqrt{g(x)}-1;\,\psi(x)),\]
formulated in an open patch of coordinates $B$, with boundary $\partial B$,
and where the $\psi(x)$ stands for a generic matter field. Consider
now the corresponding scalar-tensor theory with action \[
S_{ST}=\int d^{4}x\sqrt{g(x)}\mathcal{L}(g_{\mu\nu}(x),\phi(x);\,\psi(x)),\]
where the explicit dependency on the determinant of the metric $\sqrt{g}-1$
has been replaced by the new scalar field $\phi$. Then, the set of
(classical) solutions in both theories that fulfill that $\mathcal{L}\rightarrow0$
at the boundary $\partial B$ is the same, provided that we consider
{}``physically equivalent'' all solutions in a gauge orbit; specifically,
we take the partial gauge fixing $\sqrt{g}-1=\phi$ on the scalar-tensor
solutions:{\small \begin{eqnarray*}
\{(g_{\mu\nu},\psi)|\,\delta S_{\text{TDiff}}[g_{\mu\nu};\psi]=0,\mathcal{\, L}(g_{\mu\nu},\sqrt{g}-1;\psi)|_{\partial B}=0\} & \equiv & \left.\{(g_{\mu\nu},\phi,\psi)|\,\delta S_{ST}[g_{\mu\nu},\phi;\psi]=0,\mathcal{\, L}(g_{\mu\nu},\phi;\psi)|_{\partial B}=0\}\right|_{\phi=\sqrt{g}-1}\end{eqnarray*}
}{\small \par}
\end{thm*}
\begin{proof}
We start with a lemma
\begin{lem*}
By using the Lagrange multipliers procedure, one can establish the
following exact correspondence:\begin{equation}
\delta S_{\text{TDiff}}\equiv\left.\delta\right|_{\phi=\sqrt{g}-1}S_{ST}=0\Longleftrightarrow\delta S_{T}\equiv\delta(S_{ST}+S_{\lambda})=0\label{eq:}\end{equation}
where $S_{\lambda}\equiv\int d^{4}x\lambda(x)f[g_{\mu\nu}(x),\phi(x)]=\int d^{4}x\lambda(x)(\sqrt{g(x)}-1-\phi(x))$
and $\lambda(x)$ is taken as another dynamical field. That is to
say, the set of classical solutions of both theories, $S_{\text{TDiff}}$
and a {}``restricted $S_{ST}$'' coincide, let this be achieved
directly ---by definition--- imposing the restriction $\phi=\sqrt{g}-1$
in the configuration space, or by means of a Lagrange multiplier with
the extended action $S_{T}$.
\end{lem*}
\begin{proof}
The e.o.m.'s of the $S_{\text{TDiff}}$ theory, making use of the functional
relation $S_{\text{TDiff}}[g_{\mu\nu},\sqrt{g}-1;\,\psi]=S_{ST}[g_{\mu\nu},\phi(\sqrt{g}-1);\,\psi]$,
with $\phi(\sqrt{g}-1)= \sqrt{g}-1$, are: 
\begin{eqnarray}
0 & = & \frac{\delta S_{\text{TDiff}}}{\delta g_{\mu\nu}}\equiv\left[\frac{\delta S_{ST}}{\delta\phi}\frac{\delta\sqrt{g}}{\delta g_{\mu\nu}}+\frac{\delta S_{ST}}{\delta g_{\mu\nu}}\right]_{\phi=\sqrt{g}-1}\label{eq:}\\
0 & = & \frac{\delta S_{\text{TDiff}}}{\delta\psi}\equiv\left[\frac{\delta S_{ST}}{\delta\psi}\right]_{\phi=\sqrt{g}-1}\label{eq:}\end{eqnarray}
 to compare with the e.o.m.'s of the $S_{T}$ theory:\begin{eqnarray}
0 & = & \frac{\delta S_{T}}{\delta\lambda}\equiv\sqrt{g}-1-\phi\nonumber \\
0 & = & \frac{\delta S_{T}}{\delta\phi}\equiv\frac{\delta S_{ST}}{\delta\phi}-\lambda\label{eq:}\\
0 & = & \frac{\delta S_{T}}{\delta g_{\mu\nu}}\equiv\lambda\frac{\delta\sqrt{g}}{\delta g_{\mu\nu}}+\frac{\delta S_{ST}}{\delta g_{\mu\nu}}\nonumber \\
0 & = & \frac{\delta S_{T}}{\delta\psi}\equiv\frac{\delta S_{ST}}{\delta\psi}\nonumber \end{eqnarray}
By employing the first of these four equations in the others, and
eliminating $\lambda$ among the second and third equations, we reach
the above conclusion. 
\end{proof}
Following the main concept of this lemma, that $S_{\text{TDiff}}$ theory
is a restricted version of a $S_{ST}$ theory (previous to any variational
principle), the set of (classical) solutions of the $S_{\text{TDiff}}$ theory
would be greater in general than the normal, unrestricted $S_{ST}$
theory after gauge-fixing: this is a well-known mathematical property
of restricted variations, since the space of solutions of a restricted
theory is always bigger ---or equal--- to the unrestricted theory's
one after imposing the same condition (the two procedures, restricted
variations and restriction on the space of solutions of the unrestricted
variations, do not commute in general):\[
\{(g_{\mu\nu},\psi)|\,\delta S_{\text{TDiff}}[g_{\mu\nu};\psi]\equiv\left.\delta\right|_{\phi=\sqrt{g}-1}S_{ST}[\phi,g_{\mu\nu};\psi]=0\}\supseteq\left[\{(\phi,g_{\mu\nu},\psi)|\,\delta S_{ST}[\phi,g_{\mu\nu};\psi]=0\}\right]_{\phi=\sqrt{g}-1}\]
However, \emph{when the restriction corresponds to fixing a trivial
direction in the configuration space, the two procedures commute and
we have exactly the same set of solutions}.

The intuition comes from the case when instead of a functional $S_{ST}[\phi,g_{\mu\nu},\psi]$,
we have a function of the variables $V(\vec{x})$, say a potential,
whose extrema we want to find. Let us suppose that we have a certain
symmetry on the potential\begin{equation}
\delta_{\xi}V\equiv\delta_{\xi}\vec{x}\cdot\vec{\nabla}V=0\label{eq:symmetry_toy}\end{equation}
which we fix with a restriction $f(\vec{x})=0$ (so that $\delta_{\xi}f\equiv\delta_{\xi}\vec{x}\cdot\vec{\nabla}f\neq0$).
If we do impose the restriction before the calculus of the extrema
we will get, resorting to the Lagrange multipliers method:\begin{eqnarray}
\vec{\nabla}V+\lambda\vec{\nabla}f & = & 0\label{eq:on-shell_toy}\\
f & = & 0\nonumber \end{eqnarray}
Taking the scalar product of the first equation in the direction of
the symmetry $\delta_{\xi}\vec{x}$ we arrive to the equation\begin{equation}
\delta_{\xi}(\lambda f)\equiv\lambda\delta_{\xi}\vec{x}\cdot\vec{\nabla}f=0\label{eq:variation_of_lambdapotential_toy}\end{equation}
which indicates that $\lambda$ is $0$ on-shell and, therefore, one
could have imposed the restriction after the extremalization procedure,
just the same.

In the following, let us generalize this {}``toy-model'' to the
situation with a functional of the fields and, at the same time, particularize
to our specific case. In the first place, we require the action to
possess a symmetry under unbounded variations, $\tilde{\delta}$ (variations
that do not vanish at the boundary):\begin{equation}
\tilde{\delta}S[q_{i}(x)]\equiv\int dx[Euler-Lagrange]_{i}\tilde{\delta}q_{i}(x)+\oint d\Sigma_{\mu}\left[\frac{\partial L}{\partial(\partial_{\mu}q_{i})}\tilde{\delta}q_{i}(x)\right]_{\partial B}\label{eq:unbounded variations}\end{equation}
as opposed to the bounded variations $\delta$\begin{equation}
\delta S[q_{i}(x)]\equiv\int dx[Euler-Lagrange]_{i}\delta q_{i}(x)\label{eq:bounded variations}\end{equation}
where there is an implicit sum on the index $i$ of the various fields
$q_{i}(x)$. This symmetry is a null variation in a certain coordinate-dependent
{}``direction'' labeled $\xi$\begin{equation}
\tilde{\delta}_{\xi}S[q_{i}(x)]=0\label{eq:}\end{equation}
In passing, although the e.o.m.'s are obtained with the usual bounded
variations $\delta$, it is necessary to consider such a wider class
of variations $\tilde{\delta}$ in order to obtain the analogous of
\eqref{eq:variation_of_lambdapotential_toy}, as we will see later.
Now, since we have restricted ourselves to the case when the Lagrangian goes to zero at the boundary,
$\mathcal{L}_{ST}\underset{\partial B}{\rightarrow}0$,
we do have such a symmetry in our scalar-tensor action. For a general
infinitesimal (active) gauge transformation with spacetime parameter
$\xi^{\mu}(x)$ we obtain\begin{equation}
\left.\tilde{\delta}_{\xi}\right|_{\mathcal{L}_{ST}\underset{\partial B}{\rightarrow}0}S_{ST}=\left.\int d^{4}x\partial_{\mu}(\sqrt{g}\xi^{\mu}\mathcal{L}_{ST})\right|_{\mathcal{L}_{ST}\underset{\partial B}{\rightarrow}0}=\left.\oint d\Sigma_{\mu}[\sqrt{g}\xi^{\mu}\mathcal{L}]_{\partial B}\right|_{\mathcal{L}_{ST}\underset{\partial B}{\rightarrow}0}=0\label{eq:unbounded variation of Sst}\end{equation}
Correspondingly, in this situation we have a \emph{globally-defined}
gauge transformation that takes as from $\phi$ to $\sqrt{g}-1$ or,
what is the same, from $S_{ST}$ to $S_{\text{TDiff}}$ (the precise transformation
is detailed in Appendix C). 

In the second place, let us examine what the gauge unbounded variation
of $S_{\lambda}$ with the on-shell condition $\phi=\sqrt{g}-1$ entails
in our case:\begin{equation}
\left[\tilde{\delta}_{\xi}S_{\lambda}\right]_{\phi=\sqrt{g}-1}=\int\left[\tilde{\delta}_{\xi}g_{\mu\nu}\frac{\tilde{\delta}S_{\lambda}}{\tilde{\delta}g_{\mu\nu}}+\tilde{\delta}_{\xi}\phi\frac{\tilde{\delta}S_{\lambda}}{\tilde{\delta}\phi}\right]_{\phi=\sqrt{g}-1}=\label{eq:}\end{equation}
\begin{equation}
=\int\left[\left(\nabla_{\mu}\xi_{\nu}\right)\left(\sqrt{g}g^{\mu\nu}\lambda\right)+\left(\xi^{\mu}\partial_{\mu}\phi\right)\left(-\lambda\right)\right]_{\phi=\sqrt{g}-1}=\int\lambda\sqrt{g}\partial_{\mu}\xi^{\mu},\,\,\forall\xi^{\mu}\label{eq:intbyparts_nonpossible}\end{equation}
so that\begin{equation}
\left[\tilde{\delta}_{\xi}S_{\lambda}\right]_{\phi=\sqrt{g}-1}=0,\,\,\forall\xi^{\mu}\Longleftrightarrow\lambda=0\label{eq:condition for zero variation of Slambda}\end{equation}
in complete analogy with the situation in \eqref{eq:variation_of_lambdapotential_toy}.
Let us remark that this condition is not true when we just consider
the subset of bounded variations, so that one can perform an integration
by parts with a null boundary term in \eqref{eq:intbyparts_nonpossible}.

Finally, we need the unbounded version $\tilde{\delta}S_{T}=0$ of
the usual on-shell condition $\delta S_{T}=0$ to act as \eqref{eq:on-shell_toy}.
Here, the restriction $\mathcal{\mathcal{L}_{ST}\underset{\partial B}{\rightarrow}}0$
proves again essential. Take $\delta S_{T}=0$, so that the Euler-Lagrange
equations are satisfied, and consider the differences between \eqref{eq:unbounded variations}
and \eqref{eq:bounded variations}. The additional boundary term in
the former is zero in virtue of the restriction $\mathcal{\mathcal{L}_{ST}\underset{\partial B}{\rightarrow}}0$
and \emph{the fact that our gauge-fixing does not depend on the derivatives
of the fields}. Then, the Euler-Lagrange equations being fulfilled
in the bulk, one is left only with the contribution at the boundary
in \eqref{eq:unbounded variations}, a set of zero measure. Unless
subtleties involved, this should not contribute and we obtain:\begin{equation}
0=\left.\delta\right|_{\mathcal{L}_{ST}\underset{\partial B}{\rightarrow}0}S_{T}=\left.\tilde{\delta}\right|_{\mathcal{L}_{ST}\underset{\partial B}{\rightarrow}0}S_{T}=\left.\tilde{\delta}\right|_{\mathcal{L}_{ST}\underset{\partial B}{\rightarrow}0}S_{ST}+\tilde{\delta}S_{\lambda}.\label{eq:eom for St}\end{equation}

Combining the equations \eqref{eq:unbounded variation of Sst} and
\eqref{eq:eom for St}, with variations in a particular direction
$\xi$, we arrive to \begin{equation}
\left[\tilde{\delta}_{\xi}S_{\lambda}\right]_{\phi=\sqrt{g}-1}=0,\,\,\forall\xi^{\mu}\label{eq:}\end{equation}
Note that \eqref{eq:eom for St} implies, in particular, the on-shell
condition $\phi=\sqrt{g}-1$. Then, following \eqref{eq:condition for zero variation of Slambda}
we have $\lambda=0$ on-shell and therefore\begin{equation}
\left.\delta\right|_{\mathcal{L}_{ST}\underset{\partial B}{\rightarrow}0}S_{\text{TDiff}}\equiv\left.\delta\right|_{\mathcal{L}_{ST}\underset{\partial B}{\rightarrow}0;\,\,\phi=\sqrt{g}-1}S_{ST}=\left.\delta\right|_{\mathcal{L}_{ST}\underset{\partial B}{\rightarrow}0}S_{T}=0\Longleftrightarrow\left[\left.\delta\right|_{\mathcal{L}_{ST}\underset{\partial B}{\rightarrow}0}S_{ST}=0\right]_{\phi=\sqrt{g}-1}\label{eq:}\end{equation}

There is yet a further step that consists on proving%
\footnote{We would not need to show explicitly the corresponding statement $\left.\delta\right|_{\mathcal{L}_{ST}\underset{\partial B}{\rightarrow}0}S_{T}=0\Longleftrightarrow\left[\delta S_{T}=0\right]_{\mathcal{L}_{ST}\underset{\partial B}{\rightarrow}0}$
since the set of solutions of $\left[\delta S_{T}\equiv\left.\delta\right|_{\phi=\sqrt{g}-1}S_{ST}=0\right]_{\mathcal{L}_{ST}\underset{\partial B}{\rightarrow}0}$
is clearly contained in $\left.\delta\right|_{\mathcal{L}_{ST}\underset{\partial B}{\rightarrow}0}S_{T}=0$
and, at the same time, contains the solutions of $\left[\delta S_{ST}=0\right]_{\mathcal{L}_{ST}\underset{\partial B}{\rightarrow}0;\,\,\phi=\sqrt{g}-1}$.
Hence, it suffices to prove the equality of these two last sets.%
}\begin{equation}
\left[\left.\delta\right|_{\mathcal{L}_{ST}\underset{\partial B}{\rightarrow}0}S_{ST}=0\right]_{\phi=\sqrt{g}-1}\Longleftrightarrow\left[\delta S_{ST}=0\right]_{\mathcal{L}_{ST}\underset{\partial B}{\rightarrow}0;\,\,\phi=\sqrt{g}-1}\label{eq:}\end{equation}
or, in more generality,\begin{equation}
\left.\delta\right|_{\mathcal{L}\underset{\partial B}{\rightarrow}0}S=0\Longleftrightarrow\left[\delta S=0\right]_{\mathcal{L}\underset{\partial B}{\rightarrow}0}\label{eq:}\end{equation}
i.e., that the restriction $\mathcal{L}\underset{\partial B}{\rightarrow}0$
can be implemented equivalently before and after taking the variations.
The following lemma completes the proof.
\begin{lem*}
$\left.\delta\right|_{\mathcal{L}\underset{\partial B}{\rightarrow}0}S=0\Longleftrightarrow\left[\delta S=0\right]_{\mathcal{L}\underset{\partial B}{\rightarrow}0}$
\end{lem*}
\begin{proof}
Having a look at expression \eqref{eq:bounded variations}, one can
check that the condition $\mathcal{L}\underset{\partial B}{\rightarrow}0$
imposes no restriction to the values of the $\delta q_{i}(x)$ so
that the usual Euler-Lagrange equations are derived inside the region,
apart from the aforementioned condition. This amounts to equivalence
between both ways of calculating variations stated in the lemma.

Yet, let us impose the restriction $\mathcal{L}\underset{\partial B}{\rightarrow}0$
at the level of the action with the aid of a Lagrange multiplier.
That could be done in the following form:\begin{equation}
S_{L.m.}=\underset{\partial B}{\oint}dx_{b}\lambda(x_{b})\mathcal{L}(x_{b})=\int d^{4}x\lambda(x)\mathcal{L}(x)\delta(x-x_{b})\label{eq:}\end{equation}
\footnote{One simple picture to capture the situation is to imagine 3-dimensional
space and a boundary $S_{2}$ of radius unity. By using polar coordinates
the points $x_{b}$ would be spanned by the angular coordinates $\Omega$
(at fixed radius $r=1$). Then one would use the identity $\int dr\, r^{2}\delta(r-1)=1$
to obtain a three-dimensional integral.%
}If we now attempt to compute the e.o.m.'s corresponding to the extended
action $S_{ext}\equiv S+S_{L.m.}$, we would arrive at the expression
(see \ref{eq:unbounded variations}):\begin{eqnarray}
0=\delta S_{ext} & = & \int d^{4}x\left[[Euler-Lagrange]_{i}(x)+[Euler-Lagrange]_{i}(x)\lambda(x)\delta(x-x_{b})\right]\delta q_{i}(x)\label{eq:}\\
 &  & +\int d^{4}x\mathcal{L}(x)\delta(x-x_{b})\tilde{\delta}\lambda(x)\nonumber \end{eqnarray}
Note that, in principle, the variation of $\lambda$ does not have
to be zero at the boundary. Also, one can see that the restriction
does not lead to any modification of the equations in the bulk ($x\neq x_{b}$).
Indeed, one finally arrives to\begin{eqnarray}
[Euler-Lagrange]_{i}(x) & = & 0,\,\,\,\,\,\forall i,\, x\neq x_{b}\label{eq:}\\
\mathcal{L}(x_{b}) & = & 0\nonumber \end{eqnarray}
which is just the same set of e.o.m.'s that the expression $\left[\delta S=0\right]_{\mathcal{L}\underset{\partial B}{\rightarrow}0}$
would render.
\end{proof}
\end{proof}
Let us make some comments following the theorem. First, a totally
analogous theorem exists which connects GR and unimodular gravity
\eqref{eq:unimodular gravity}. Indeed, the action for unimodular
gravity can be written with the help of a Lagrange multiplier as in
Appendix B ($\epsilon_{0}(x)=1$)\begin{equation}
S_{UG}=\int d^{4}x\left[\sqrt{g(x)}R(x)+\lambda(x)(\sqrt{g(x)}-1)\right]\label{eq:unimodular gravity with Lagrange multiplier}\end{equation}
and, following a similar reasoning, one can establish the equivalence
of the solutions of both theories when the Lagrangian goes to zero
at the boundary. This result is already well-known, since the integration
constant playing the role of a cosmological constant that appears
in unimodular gravity --- with respect to GR --- is non-zero at the
boundary. 

Actually, a far more general conclusion can be drawn from our line
of reasoning: \emph{for any Lagrangian theory with a globally-defined
gauge symmetry, it is equivalent to impose a gauge-fixing without
derivatives before or after the calculation of variations, provided
that we only consider (classical) solutions that make the Lagrangian
vanish at the boundary of spacetime}. This adds up to the result in
\cite{Pons:1995ss,Pons:2009ch} (Hamiltonian and Lagrangian treatment,
respectively), that {}``the effect of plugging the gauge fixing constraint
into the Lagrangian can be compensated by adding to the equations
of motion for the reduced {[}restricted] theory some constraints that
have disappeared as such along the process''%
\footnote{The paradigmatic example is the temporal gauge $A_{0}=0$ of electromagnetism
which, if plugged directly in the action before the calculus of variation,
leads to the disappearance of the Gauss constraint, $\partial_{\alpha}F^{\alpha0}=0$,
although it is kept compatible with the dynamical (Hamiltonian-picture)
evolution. Our boundary condition on the Lagrangian restores it. On
the other hand, the Lorenz gauge $\partial_{\mu}A^{\mu}=0$, including
a derivative, escapes our conclusion.%
}. 

Second, this theorem somehow contradicts common beliefs of other authors. On the other hand, the results in reference \cite{Pirogov:2009hr} are compatible with the theorem; this latter work is a previous study on the relation between TDiff (called "unimodular metagravity") and scalar-tensor from another perspective\footnote{In \cite{Pirogov:2009hr}, the author finds solutions for the TDiff theory (referred
to as \char`\"{}unimodular metagravity\char`\"{}) that are not present
in the \emph{corresponding} scalar-tensor theory (corresponding in
the sense stipulated by our theorem). These solutions are parametrized
by a non-zero value of the parameter $\lambda_{h}$ in equation (9).
According to our theorem, they should all imply a non-zero value of
the Lagrangian at the boundary of spacetime. We are not going to find
all the solutions, but we note that the trivial configuration ($g_{\mu\nu}\rightarrow\eta_{\mu\nu}$,
$\chi\rightarrow0$, $\Lambda=0$) at the boundary, which makes the
Lagrangian vanish, is a solution only for $\lambda_{h}=0$.

Regarding the $\lambda_{h}=0$ case, the author has confirmed the equivalence between both theories 
[private communication].%
}.

Third, we have expressed the explicit dependence on $g$ of the TDiff
theory by $\sqrt{g}-1$ throughout the paper, so that it vanishes
in Minkowski space with (canonical) cartesian coordinates. Whether
the requirement $\mathcal{L}\underset{\partial B}{\rightarrow}0$
is {}``natural'' or too restrictive is left to the reader.

Fourth and last, we have used a Lagrangian formalism throughout the proof. It would be desirable
to complement this point of view by resorting to the Hamiltonian formalism of constrained systems, as in 
\cite{Henneaux:1989zc} for the case of unimodular gravity.	 

\section{Energy-Momentum Tensors}

In this section, we want to compare the energy-momentum tensors naturally defined in both theories: scalar-tensor and TDiff. Are these the same quantity? how do they relate to each other? For the comparison, we will use a "gauge-fixed" scalar-tensor, although we know that this is physically-equivalent to the covariant one, as was stated in the introduction.

The Rosenfeld prescription for the energy-momentum tensor\begin{equation}
T^{\mu\nu}\equiv\frac{1}{\sqrt{g}}\frac{\delta S_{M}}{\delta g_{\mu\nu}}\label{eq:}\end{equation}
(where $S_{M}$ stands for the matter action) guarantees that this
tensor is conserved on-shell, $\nabla_{\mu}T^{\mu\nu}=0$,\emph{ in
any Lagrangian-based generally covariant metric theory} \emph{without
absolute non-dynamical variables,} through the gravitational-field equations \cite{LeeConservlaws:1974nq}
(and also directly, by the definition of the matter Lagrangian in
a metric theory itself); this is the case of a metric scalar-tensor
theory. Let's compare what the Rosenfeld prescription entails for
both theories, {}``gauge-fixed scalar-tensor'' and TDiff:

\begin{equation}
\left.T_{ST}^{\mu\nu}\right|_{\phi=\sqrt{g}-1}\equiv\frac{1}{\sqrt{g}}\left[\frac{\delta S_{M,ST}}{\delta g_{\mu\nu}}\right]_{\phi=\sqrt{g}-1}\label{eq:energymomentum tensor ST}\end{equation}

\begin{equation}
T_{\text{TDiff}}^{\mu\nu}\equiv\frac{1}{\sqrt{g}}\frac{\delta S_{M,\text{TDiff}}}{\delta g_{\mu\nu}}=\frac{1}{\sqrt{g}}\frac{\left.\delta\right|_{\phi=\sqrt{g}-1}S_{M,ST}}{\delta g_{\mu\nu}}\label{eq:energymomentum tensor TDiff}\end{equation}

A rough summary of the theorem in the previous section would be that,
under certain boundary conditions, the restriction $\phi=\sqrt{g}-1$
can be taken before or after the variations in the action:\begin{equation}
\left[\delta S_{M,ST}\right]_{\phi=\sqrt{g}-1}\sim\left.\delta\right|_{\phi=\sqrt{g}-1}S_{M,ST}\label{eq:}\end{equation}
However, note that this equality does not extend to the energy-momentum
tensors as previously defined:\begin{equation}
\left[\frac{1}{\sqrt{g}}\frac{\delta S_{M,ST}}{\delta g_{\mu\nu}}\right]_{\phi=\sqrt{g}-1}\neq\frac{1}{\sqrt{g}}\frac{\left[\delta S_{M,ST}\right]_{\phi=\sqrt{g}-1}}{\delta g_{\mu\nu}}=\frac{1}{\sqrt{g}}\frac{\left.\delta\right|_{\phi=\sqrt{g}-1}S_{M,ST}}{\delta g_{\mu\nu}}\label{eq:}\end{equation}
Indeed, a direct computation with the generic matter action $S_{M,ST}[\phi,g_{\mu\nu};\,\psi]$
gives us\begin{equation}
T_{\text{TDiff}}^{\mu\nu}=\frac{1}{\sqrt{g}}\left[\frac{\delta S_{M,ST}}{\delta g_{\mu\nu}}\right]_{\phi=\sqrt{g}-1}+\frac{1}{\sqrt{g}}\left[\frac{\delta S_{M,ST}}{\delta\phi}\right]_{\phi=\sqrt{g}-1}\frac{1}{2}\sqrt{g}g^{\mu\nu},\label{eq:}\end{equation}
to be compared with \eqref{eq:energymomentum tensor ST}, and note
that the second term does not have to vanish in general. We thus have
\emph{two different} \emph{definitions} of the energy-momentum tensor.

In conclusion, the strict use of the Rosenfeld prescription in both theories does lead to different energy-momentum tensors, even if the theories themselves are equivalent: the Rosenfeld energy-momentum tensor depends on the representation of the theory. Once one accepts the correspondence between scalar-tensor and TDiff theory, there is not really a compelling reason to use $T_{\text{TDiff}}^{\mu\nu}$, which does not share with $T_{ST}^{\mu\nu}$ the property of being conserved.

\section{Flat Limit}

In this section, we want to compare the (two different kinds of) flat limits in both theories: scalar-tensor and TDiff. We should and, indeed, we do obtain the same behaviour, in agreement with their established equivalence. The (false) impression that a flat limit of TDiff should just amount to the substitution of the determinant of the metric $\sqrt{g}$ by unity, is deceiving. 

\subsection{Global flat limit (weak field)}

The (global) flat limit is defined by $\mathbf{\boldsymbol{g}}\rightarrow\boldsymbol{\xi}$
with $R[\boldsymbol{\xi}]=0$, i.e., the Riemann curvature tensor
of the metric tensor vanishes. In coordinates: $g_{\mu\nu}(x)\underset{g.f.l.}{\rightarrow}\xi_{\mu\nu}(x)$,
with $R_{\mu\nu\lambda\gamma}[\xi](x)=0$. The Minkowski form of the
metric ($\eta_{\mu\nu}(x)=diag(-1,1,1,1)$) is just one member of
the set of flat coordinate metrics: $\eta_{\mu\nu}(x)\in\{\xi_{\mu\nu}(x)\}$.

\begin{itemize}
\item Scalar-tensor:

$S_{M}=\int dV\sqrt{g}f(\phi)(\frac{1}{2}g^{\mu\nu}\partial_{\mu}\Psi\partial_{\mu}\Psi-V(\Psi))\underset{g.f.l.}{\rightarrow}S_{M}=\int dV\sqrt{\xi}f(\phi)(\frac{1}{2}\xi^{\mu\nu}\partial_{\mu}\Psi\partial_{\mu}\Psi-V(\Psi))=$\\
$\underset{Diff\, invariance}{=}\int dV\sqrt{\eta}f(\phi)(\frac{1}{2}\eta^{\mu\nu}\partial_{\mu}\Psi\partial_{\mu}\Psi-V(\Psi))=\int dV\,\, f(\phi)(\frac{1}{2}\eta^{\mu\nu}\partial_{\mu}\Psi\partial_{\mu}\Psi-V(\Psi))$

\item TDiff:\\
We formulate the action in a certain set of coordinates (characterized
by the condition that the $J=1$).

$S_{M}=\int dV\sqrt{g}f(\sqrt{g}-1)(\frac{1}{2}g^{\mu\nu}\partial_{\mu}\Psi\partial_{\mu}\Psi-V(\Psi))\underset{g.f.l.}{\rightarrow}S_{M}=\int dV\sqrt{\xi}f(\sqrt{\xi}-1)(\frac{1}{2}\xi^{\mu\nu}\partial_{\mu}\Psi\partial_{\mu}\Psi-V(\Psi))=$\\
$\underset{coord.\, change}{=}\int dV\sqrt{\eta}f(J^{-1}\sqrt{\eta}-1)(\frac{1}{2}\eta^{\mu\nu}\partial_{\mu}\Psi\partial_{\mu}\Psi-V(\Psi))=\int dV\,\, f(J^{-1}\sqrt{\eta})(\frac{1}{2}\eta^{\mu\nu}\partial_{\mu}\Psi\partial_{\mu}\Psi-V(\Psi))$\\
A Jacobian of the coordinate transformation with respect to the original
set of coordinates, $J$, has appeared in a general set of coordinates. The quantity $J^{-1}\sqrt{\eta}-1$
behaves as a (Diff) scalar. If we set its value to a constant in a
certain set of coordinates, it will be the same constant also in every
other set of coordinates / reference frame. Therefore, much as in
scalar-tensor theory, we can not really get rid of it, unless we want
to make the whole theory trivial. 

\end{itemize}

\subsection{Local flat limit ({}``free fall'').}

It is always possible to find a coordinate system where the first
derivatives of the metric --- or equivalently, the Christoffels ---
vanish along a worldline, which we can take timelike $\mathcal{P}(\tau)$,
such that in the new coordinate system $g_{\mu\nu}\rightarrow g_{\hat{\mu}\hat{\nu}}$
one has the {}``free fall'' approximation\[
g_{\hat{\mu}\hat{\nu}}(x)=\eta_{\hat{\mu}\hat{\nu}}(\mathcal{P}(\tau))+o((x-\mathcal{P}(\tau))^{2}\partial^{2}g_{\hat{\mu}\hat{\nu}})\]

If we start from a TDiff action $S_{M}=\int dV\sqrt{g}f(\sqrt{g}-1)(\frac{1}{2}g^{\mu\nu}\partial_{\mu}\Psi\partial_{\mu}\Psi-V(\Psi))$
and change to the {}``caret'' coordinates, we will introduce Jacobians:\[
S_{M}=\int dV\sqrt{g}f(\sqrt{g}-1)(\frac{1}{2}g^{\mu\nu}\partial_{\mu}\Psi\partial_{\mu}\Psi-V(\Psi))\underset{coord.\, change}{=}\int dV\sqrt{\hat{g}}f(J^{-1}\sqrt{\hat{g}}-1)(\frac{1}{2}\hat{g}^{\mu\nu}\partial_{\mu}\Psi\partial_{\mu}\Psi-V(\Psi))\]
So, when we finally do a Taylor-like expansion around $\mathcal{P}(\tau)$
(we take the fields to {}``live'' in a local neighborhood of the
timelike worldline)\[
S_{m}\simeq\int dV\sqrt{\hat{\eta}}f(J^{-1}\sqrt{\hat{\eta}}-1)(\frac{1}{2}\hat{\eta}^{\mu\nu}\partial_{\mu}\Psi\partial_{\mu}\Psi-V(\Psi))=\int dV\,\, f(J^{-1}\sqrt{\hat{\eta}}-1)(\frac{1}{2}\hat{\eta}^{\mu\nu}\partial_{\mu}\Psi\partial_{\mu}\Psi-V(\Psi))\]
the scalar $J^{-1}\sqrt{\hat{\eta}}-1$ is present, much like the
corresponding $\phi$ in scalar-tensor theory.

\section{Conclusions}

Let us summarize our results. We started by imposing a TDiff gauge
symmetry group to the gravitational Lagrangian based on a symmetric
rank-two tensor $g_{\mu\nu}$, as opposed to the standard Diff symmetry
that leads us to General Relativity. Since the symmetry group is smaller,
we thus get to a more general Lagrangian \eqref{eq:TDiff action}\eqref{eq:TDiff matter action}.
Then, the theorem stated in section 3 tells us that, modulo some boundary
conditions, one can view this theory just as a general scalar-tensor
theory that has been presented in a particular gauge, namely $\phi=\sqrt{g}-1$,
but whose inherent symmetry group is again Diff (it is {}``physically-equivalent'',
in the sense of the introduction). At the end of the day, we have
{}``gained'' a scalar degree of freedom that was not postulated,
but appears automatically due to the lack of symmetry. In a way,\emph{
the TDiff group can be thus considered to give rise to general-scalar
tensor theory}. This present knowledge builds up from previous works
\cite{AlvarezAntonVillarejo:2008zw,AlvarezAntonVillarejo:2009ga}.
In addition, some collateral consequences of the theorem with a wider
scope of application are presented at the end of section 3: these
confront the question of {}``whether it makes a difference to impose
the gauge-fixing before or after the calculus of variations'' 

Regarding the observational signatures of TDiff theories at the classical
level, one can therefore rely on this equivalence to take advantage
from the abundant literature that studies the bounds on the scalar-tensor
theory. The same metric / universal coupling requirement that is usually
imposed for scalar-tensor theory --- which was mentioned in the introduction
--- can be applied directly in the TDiff Lagrangian \eqref{eq:TDiff action}\eqref{eq:TDiff matter action}
by considering only matter Lagrangians without explicit dependence
on $g$.

\begin{acknowledgement*}
I am grateful to Enrique Alvarez and Anton F. Faedo for our previous
work on the subject, that opened this interesting perspective. I have
largely benefited from discussion with Josep M. Pons, and from his
comments. I also want to thank Daniel Zenh\"{a}usern for his interest and help with the project during my stay at EPFL, Lausanne. On the other hand, I missed some more support from the mathematics
community that I consulted. This work has been partially supported by the European
Commission (HPRN-CT-200-00148) as well as by FPA2009-09017 (DGI del
MCyT, Spain) and S2009ESP-1473 (CA Madrid). The author is personally
supported by a MEC grant, AP2007-00385.

\end{acknowledgement*}
\appendix

\section{Active symmetry group vs. passive covariance group.}

We want to distinguish two different ways in which the diffeomorphism
group (Diff or TDiff) enters our considerations, although this separation
is not standard: there is the concept of the (passive) covariance
group of the representation, an invariance of form of the theory's
Lagrangian when the spacetime coordinates are transformed; then, there
is also the associated {}``active covariance group'', which one
can consider a gauge symmetry group and is just the active version
of the former transformation on the dynamical fields (not the background
quantities) that leaves the action invariant. Both groups do not necessarily
have to coincide, and the difference lies precisely in our treatment
of the background quantities.

For example, the quantity \begin{equation}
\int d^{4}x\sqrt{g}g\label{eq:}\end{equation}
is both (just) TDiff covariant with respect to coordinate transformations
and has (just) a TDiff gauge symmetry group. However, the quantity

\begin{equation}
\int d^{4}x\sqrt{g}\frac{g}{\tilde{g}},\label{eq:}\end{equation}
which includes a background density field $\tilde{g}$, is fully generally
covariant with respect to coordinate transformations, while its gauge
symmetry group remains (just) TDiff, due to the presence of a fixed
absolute prior-geometric object.

As can be inferred from the example, the gauge symmetry implies the
covariance under the associated coordinate transformation, but not
vice versa.

\section{Unimodular gravity vs General Relativity}

It is widely known that unimodular theories are classically equivalent
to General Relativity when the latter is {}``gauge-fixed'' with
the condition%
\footnote{It does not matter whether we take $\sqrt{g}$ or $g$ to 1 \cite{BijDamNg:1981ym}%
} $g=1$, safe for an arbitrary cosmological constant term that emerges
in the form of an integration constant. One can immediately see this
by using the Lagrange multiplier procedure to implement the constraint
on the determinant of the metric:

\begin{equation}
S_{UG}\equiv-\frac{1}{2k^{2}}\int d^{4}x\left[\sqrt{g}R+\sqrt{g}L_{SM}\left[\psi_{m},g_{\mu\nu}\right]\right]_{\sqrt{g}=1}\label{eq:}\end{equation}
\begin{equation}
\Longrightarrow S_{UG,ext}\equiv S_{GR}+S_{UG,\lambda}=-\frac{1}{2k^{2}}\int d^{4}x\left[\sqrt{g}R+\sqrt{g}L_{SM}\left[\psi_{m},g_{\mu\nu}\right]\right]+\int d^{4}x\lambda(\sqrt{g}-1)\label{eq:}\end{equation}
As compared to the GR template, the only difference is the appearance
in the e.o.m for $g_{\mu\nu}$,$\frac{\delta S_{UG,\lambda}}{\delta g_{\mu\nu}}=0$,
of an extra term of the type $\sim\sqrt{g(x)}\lambda(x)g^{\mu\nu}(x)$
; furthermore, the (on-shell) Bianchi identity stemming from the Diff-invariant
terms of the action forces $\lambda(x)=\lambda=const.$ Then, the
equation for $\lambda$, $\frac{\delta S_{UG,\lambda}}{\delta\lambda}=0$,
fixes the \char`\"{}gauge\char`\"{} $\sqrt{g}=1$.

\section{Gauge transformation between $\phi$ and $\sqrt{g}-1$}

Under an (active) coordinate gauge transformation, we have the transformation
rules:\begin{equation}
\sqrt{g'(x)}-1=J^{-1}\sqrt{g(x)}-1+\left[\sqrt{g'(x)}-\sqrt{g'(x')}\right],\,\, with\,\,\,\sqrt{g'(x')}=J^{-1}\sqrt{g(x)},\label{eq:}\end{equation}
\begin{equation}
\phi'(x)=\phi(x)+\left[\phi'(x)-\phi'(x')\right],\,\, with\,\,\,\phi'(x')=\phi(x),\label{eq:}\end{equation}
which correspond to the Lie derivative at the infinitesimal level,\begin{equation}
\sqrt{g'(x)}-1=\sqrt{g(x)}-1+\sqrt{g(x)}\partial_{a}\xi^{a}(x)+\xi^{a}(x)\partial_{a}\sqrt{g(x)}=\sqrt{g(x)}+\partial_{a}(\sqrt{g(x)}\xi^{a}(x))=\sqrt{g(x)}+\sqrt{g(x)}\nabla_{a}\xi^{a}(x)\label{eq:}\end{equation}
\begin{equation}
\phi'(x)=\phi(x)+\xi^{a}(x)\partial_{a}\phi(x)\label{eq:}\end{equation}
We see that taking the jacobian $J(x)=\sqrt{g(x)}/\left(\phi(x)+1\right)$
we will obtain equality among $\phi'$ and $\sqrt{g'}-1$ at each
point of spacetime (the {}``transport terms'' in squared brackets
will be equal).\\
\\

\end{document}